\documentclass[11pt]{article}

\usepackage[utf8]{inputenc}

\usepackage{amsmath,amsfonts,amsthm,amssymb,color}
\usepackage[dvipsnames,svgnames,table]{xcolor}
\definecolor{darkgreen}{rgb}{0.0,0,0.9}
\usepackage{mathtools}
\usepackage{authblk}
\usepackage{fullpage}
\usepackage{thm-restate}
\usepackage{parskip}
\usepackage{comment}
\usepackage{tikz}
\usetikzlibrary{arrows.meta}
\usepackage{bbm}
\usepackage{float}
\usepackage{enumitem}
\usepackage{dsfont}
\usepackage[sc]{mathpazo}
\usepackage[basic]{complexity}
\usepackage{todonotes}
\presetkeys{todonotes}{inline}{}
\usepackage{algorithm2e}
\usepackage[colorlinks=true,citecolor=OliveGreen,linkcolor=BrickRed,urlcolor=BrickRed,pdfstartview=FitH]{hyperref}
\usepackage[capitalize,nameinlink]{cleveref}
\usepackage{tcolorbox}
\usepackage{bbm}
\usepackage{algpseudocode}
\usepackage{xcolor}
\usepackage{amsmath}
\newcommand{\highlight}[1]{\colorbox{yellow!30}{$\displaystyle #1$}}

\newtcolorbox{wbox}
{
	colback  = white,
}

\SetKwInOut{Input}{Input}
\SetKwInOut{Output}{Output}
\SetKwBlock{InParallel}{in parallel do}{end}
\SetKwFor{ParallelFor}{for}{in parallel do}{end}

\theoremstyle{definition}
\newtheorem{theorem}{Theorem}
\newtheorem{lemma}{Lemma}
\newtheorem{corollary}{Corollary}
\newtheorem{definition}{Definition}

\newtheorem{observation}[theorem]{Observation}

\newtheorem{remark}[theorem]{Remark}
\newtheorem{example}{Example}

\newcommand{\calA}{\mathcal{A}}
\newcommand{\calR}{\mathcal{R}}

\title{Fair Rent Division: New Budget and Rent Constraints}

 \author[1]{Rohith Reddy Gangam}
 \author[1]{Shayan Taherijam}
 \author[1]{Vijay V. Vazirani}

 \affil[1]{University of California, Irvine}

\date{}

\begin{document}
\maketitle

\begin{abstract}
We study the classical \emph{rent division problem}, where $n$ agents must allocate $n$ indivisible rooms and split a fixed total rent $R$. The goal is to compute an \emph{envy-free} (EF) allocation, where no agent prefers another agent’s room and rent to their own. This problem has been extensively studied under standard assumptions, where efficient algorithms for computing EF allocations are known.

We extend this framework by introducing two practically motivated constraints: \emph{(i)} lower and upper bounds on room rents, and \emph{(ii)} room-specific budget for agents. We develop efficient combinatorial algorithms that either compute a feasible EF allocation or certify infeasibility.

We further design algorithms to optimize over EF allocations using natural fairness objectives such as \emph{maximin utility}, \emph{leximin utility}, and \emph{minimum utility spread}. Our approach unifies both constraint types within a single algorithmic framework, advancing the applicability of fair division methods in real-world platforms such as \emph{Spliddit}.
\end{abstract}

\section{Introduction}
\label{sec:intro}
Fair rent division is a well-studied problem in economics and algorithmic game theory, and has found practical deployment through platforms such as \emph{Spliddit} and broad public interest via the \emph{New York Times}' interactive coverage. The classical problem involves assigning $n$ agents to $n$ indivisible rooms and splitting a fixed total rent $R$ such that no agent envies another's room and rent. That is, each agent strictly prefers their own room and rent over any other's. An allocation satisfying this property is called \emph{envy-free} (EF).

Under standard assumptions—quasi-linear utilities and no additional constraints—EF allocations are guaranteed to exist and can be computed efficiently via combinatorial or LP-based methods~\cite{svensson1983large, klijn2000}. In this work, we move beyond these assumptions and study the computation of EF allocations under two realistic extensions: bounds on individual room rents and room-specific budgets for agents.

Rent bounds on rooms arise naturally in many real-world housing arrangements. Some rooms may have contractual lower or upper limits on rent—either because landlords impose minimum rents for premium rooms or because roommates collectively agree to cap rents to maintain affordability. For instance, universities that partner with private landlords near campus must allocate housing fairly among students while honoring the minimum contractual rents specified by landlords. Such bounds capture practical and institutional constraints that cannot be represented by valuations alone, and thus warrant explicit modeling in the framework.

Our second extension considers budget constraints at the granularity of agent-room pairs. That is, each agent specifies a distinct upper bound on the rent they are willing to pay for each room. This generalizes prior models such as~\cite{Procaccia_Velez_Yu_2018}, which assume a single budget per agent. Room-specific budgets capture a variety of real-world financial considerations that valuations alone cannot express. Factors such as room-dependent utility surcharges can impose affordability constraints that vary across rooms for each agent. While valuations encode preferences, they do not reflect these heterogeneous financial limits. Incorporating agent-room-specific budgets allows us to retain expressive utilities while modeling realistic payment constraints.

We develop efficient combinatorial algorithms that either find an envy-free allocation satisfying the constraints or certify that none exists. We also provide efficient algorithms that identify EF allocations that are fair according to maximin utility, leximin utility, and minimum utility spread (both absolute and relative). Among the potentially many EF allocations, these fairness notions help identify outcomes that are more equitable and have been well studied in the literature on fair division. Finally, we give a unified algorithm that simultaneously handles both types of constraints—room rent bounds and room-specific budgets.

\textbf{Organization.} \Cref{sec:related} reviews related work. In \Cref{sec:preliminaries}, we present the necessary preliminaries. \Cref{sec:rent_bounds} introduces an algorithm for computing envy-free allocations under bounds on room rents, along with algorithms for selecting fair allocations using maximin, leximin, and minimum spread criteria. In \Cref{sec:agent_budgets}, we give an algorithm for handling room-specific budget constraints for agents. Finally, in \Cref{sec:rent_bounds_budgets}, we combine both types of constraints and provide an algorithm that simultaneously satisfies them. Proofs of theorems marked $\dagger$ are provided in the \Cref{app:proofs}. 

\section{Related Work}
\label{sec:related}

The rent division problem was introduced in the economics literature by Foley~\cite{foley1967resource}, and further formalized by Alkan, Demange, and Gale~\cite{adg91}. The envy-free (EF) solution was shown to always exist under quasi-linear utilities without constraints Svensson~\cite{svensson1983large}, and Klijn~\cite{klijn2000} gave a strongly polynomial algorithm based on envy graphs.

More recently, Procaccia et al.~\cite{gal2016fairest} proposed a maximin envy-free solution and implemented it as the rent division algorithm on \emph{Spliddit}. Their method optimizes fairness among EF solutions via linear programming. In follow-up work, Procaccia et al.~\cite{Procaccia_Velez_Yu_2018} studied rent division with agent budgets, but assumed each agent has a single budget across rooms.

Our work generalizes this model by introducing \emph{agent-room-specific budgets}, which better capture real-world constraints such as utility thresholds and room-specific surcharges. We provide combinatorial procedures for computing fair envy-free allocations under these constraints.

While our focus is on monetary envy-free rent division, there is significant literature on non-monetary settings. Surveys such as by Amanatidis et al.~\cite{survey_fair_division} review recent progress on discrete fair division problems involving indivisible goods, and notions like EF1, EFX, and MMS. Another relevant survey by Salvatierra et al.~\cite{survey_ef_pricing} explores envy-free pricing problems in various economic settings, emphasizing the computational complexity and algorithmic aspects of maximizing seller revenue under envy-freeness constraints.

Extensions to rent division under complex settings have also been considered, including online arrivals~\cite{aleksandrov2020online}, multi-apartment rent division~\cite{procaccia2024multiapartment} and continuous, monotone decreasing, and piecewise-linear utilities \cite{Arunachaleswaran_rent_division}. The practical relevance of this problem is evident from platforms like \emph{Spliddit} and its popular coverage in the \emph{New York Times}~\cite{nyt14}.


\section{Preliminaries}
\label{sec:preliminaries}
\subsection{Rent Division Problem}
We consider a rent division problem involving a set of $n$ agents $\mathcal{A} = \{1, 2, \dots, n\}$ and a set of $n$ rooms $\mathcal{R} = \{1, 2, \dots, n\}$. Each room must be assigned to exactly one agent, and the room rents must be such that the total rent collected equals a fixed amount $R_t$. Each agent $i \in \mathcal{A}$ has a non-negative valuation $v_{ij} \in \mathbb{R}_+$ for room $j \in \mathcal{R}$, indicating their subjective worth of the room. These valuations are encoded in a valuation matrix $V = [v_{ij}] \in \mathbb{R}_+^{n \times n}$. A rent division instance is thus defined by the tuple $I = (\mathcal{A}, \mathcal{R}, V, R_t)$.

An \emph{assignment} of rooms is a bijection $\sigma: \mathcal{A} \to \mathcal{R}$ and we will use $\sigma(i), i \in \calA$ to denote the room assigned to agent $i$ and abuse the notation by also using $\sigma(j), j \in \calR$ to denote the agent assigned to room $j$. A \emph{rent division} is specified by a vector $r \in \mathbb{R}^n$, where $r_j$ denotes the rent charged for room $j$ and $\sum_{j\in \calA} r_j = R_t$.

Given a rent vector $r$, the \emph{quasi-linear utility} of agent $i$ for room $j$ is defined as $ u_{ij}(r) = v_{ij} - r_j$. When the rent vector is clear from context, we write $u_{ij}$ instead of $u_{ij}(r)$. Under an assignment $\sigma$, agent $i$ is said to \emph{envy} agent $j$ if $v_{i\sigma(i)} - r_{\sigma(i)} < v_{i\sigma(j)} - r_{\sigma(j)}$.

A rent allocation $(\sigma, r)$ consisting of an assignment and a rent division is said to be \emph{envy-free} (EF) if no agent envies any other, i.e.,
$$
\forall i, j \in \mathcal{A},\quad u_{i\sigma(i)} \geq u_{i\sigma(j)}
\quad \text{or equivalently}\quad 
\forall i \in \mathcal{A},j\in \mathcal{R} \quad u_{i\sigma(i)} \geq u_{ij}
$$

\subsection{Welfare-maximizing allocations}

\begin{definition}
    An assignment $\sigma$ is said to be \emph{welfare-maximizing} if it maximizes the total valuation, i.e., $\sum_{i \in \mathcal{A}} v_{i\sigma(i)}.$    
\end{definition}

The following well-known result connects welfare-maximizing assignments with envy-free solutions(for reference, see \cite{mas-colell.whinston.ea95}).

\begin{theorem}[First and Second Welfare Theorems]\label{thm:welfare}
Let $(\sigma, r)$ be an envy-free (EF) allocation under instance $I$. Then:
\begin{enumerate}
    \item $\sigma$ is welfare-maximizing.
    \item For any other welfare-maximizing assignment $\sigma'$, the allocation $(\sigma', r)$ is also envy-free.
\end{enumerate}
\end{theorem}


\begin{corollary}
\label{cor:equal_utility}
If both $(\sigma, r)$ and $(\sigma', r)$ are envy-free, then each agent receives the same utility in both allocations:
$$
v_{i\sigma(i)} - r_{\sigma(i)} = v_{i\sigma'(i)} - r_{\sigma'(i)} \quad \text{for all } i \in \mathcal{A}.
$$
\end{corollary}

Given the result above, the problem of finding an envy-free allocation can be naturally interpreted using a graph-theoretic framework explained below.

\subsection{Envy Graph}

\begin{definition}
    Let $I = (\mathcal{A}, \mathcal{R}, V, R_t)$ be a rent-division instance, and let $(\sigma, r)$ be an allocation of rooms and rents. Define the \emph{envy graph} $G_E(\sigma, r) = (\mathcal{R}, E)$ on the set of rooms as follows: There is a directed edge from room $i$ to room $j$ if the agent assigned to room $i$ weakly envies the agent assigned to room $j$ under rent vector $r$. That is,
    $$
    E = \{ (\sigma(i), \sigma(j)) \mid u_{i\sigma(i)} \leq u_{i\sigma(j)} \}
    $$
\end{definition}

\begin{definition}
    An edge $(i, j) \in E$ is called a \emph{weak envy edge} if $u_{i\sigma(i)} = u_{i\sigma(j)}$, and a \emph{strong envy edge} if $u_{i\sigma(i)} < u_{i\sigma(j)}$.
\end{definition}

This combinatorial structure has been widely utilized in the literature on envy-free rent division. In particular, Klijn~\cite{klijn2000} employed it to design a strongly polynomial-time combinatorial algorithm for computing envy-free allocations. For more details check \Cref{subsec:Klijn} 


\subsection{Model}

We now extend the rent division problem to incorporate the two additional constraints of our setting. First, each agent is subject to a room-specific budget constraint, represented by a matrix $b = (b_{ij})_{(i, j) \in \mathcal{A} \times \mathcal{R}}$, where $b_{ij}$ denotes the maximum rent agent $i$ is willing or able to pay for room $j$. Second, we are given vectors $l,u \in \mathbb{R}^n$ such that each room has associated lower and upper bounds on its rent: for each room $j \in \mathcal{R}$, the rent must satisfy $l_j \leq r_j \leq u_j$. 

A rent division instance with these extended constraints is defined by the tuple $I = (\mathcal{A}, \mathcal{R}, V, R_t, b, l, u)$, and in the case that there is no budget constraint or bound, we drop them from the notation. The goal is to compute an envy-free allocation that satisfies all constraints and, when possible, optimizes one of the following fairness objectives:
\begin{itemize}
    \item \emph{Maximin utility}: maximize the minimum utility among agents;
    \item \emph{Leximin utility}: lexicographically maximize the vector of utilities;
    \item \emph{Minimum absolute spread}: minimize the difference between the maximum and minimum utilities;
    \item \emph{Minimum relative spread}: minimize the ratio of the maximum to minimum utilities.
\end{itemize}

We would like to note that the results in \Cref{sec:rent_bounds} can be obtained using Linear Programming - an optimal matching can be found separately and the conditions for envy-free rents can be written as linear constraints. This enables us to write the problem as an LP and find an envy free rent assignment that maximizes one of the above fairness objectives. The primary goal of the paper is to instead find these assignments using combinatorial means and in doing so, develop structural insights into the problem. Moreover, if the agents have room-specific bounds, these conditions can not be written as linear constraints and so, we will use these combinatorial insights to solve the problem in \Cref{sec:agent_budgets} and \Cref{sec:rent_bounds_budgets}.

We first begin by analyzing the setting with agent budgets but without rent bounds.

\section{Envy-free allocation with bounds on room rents}
\label{sec:rent_bounds}

We begin by observing that the First and Second Welfare Theorems continue to hold even under these additional constraints.

\begin{restatable}{theorem}{welfarebounds}
\label{thm:welfare_bounds}
$^{\dagger}$
Let $(\sigma, r)$ be an envy-free allocation under instance $I$. Then:
\begin{enumerate}
    \item $\sigma$ is a welfare-maximizing assignment.
    \item For any other welfare-maximizing assignment $\sigma'$, the allocation $(\sigma', r)$ is also envy-free.
\end{enumerate}
\end{restatable}

This result simplifies the problem: it suffices to first compute a welfare-maximizing assignment, and then find rents (within the given bounds) that make the allocation envy-free.

Expanding on these properties, we will describe \Cref{alg:bounds_on_rents} to obtain an envy-free rent satisfying the bounds on room rents. We will then provide algorithms to obtain fair envy-free rents in Sections \ref{sec:maximin_and_leximin}, and \ref{sec:min_spread}.

Let $I = (\mathcal{A}, \mathcal{R}, V, R_t, l, u)$ be a rent division instance with lower and upper bounds on the rents of each room. By \Cref{thm:welfare_bounds}, we may begin with any welfare-maximizing assignment and search for envy-free rents that respect the specified bounds. \Cref{alg:bounds_on_rents} describes an algorithm that, given a welfare-maximizing assignment $\sigma$ and an envy-free rent vector $r$ (e.g., computed via an algorithm such as \cite{klijn2000}), computes an envy-free allocation satisfying all rent bounds, or correctly reports that there is no solution for this instance.

We begin with an observation about the envy graph $G_E(\sigma, r)$ that is an integral part of all our algorithms. Consider an edge $(i,j)$ in $G_E(r)$, indicating that agent $i$ weakly envies agent $j$. If we decrease the rent of room $j$, then to maintain envy-freeness, the rent of room $i$ must decrease by at least the same amount—otherwise, the weak envy becomes strong. Similarly, increasing the rent of room $i$ requires increasing the rent of room $j$ by at least as much. This implies that increasing (resp., decreasing) the rent of a room requires increasing (resp., decreasing) the rents of all rooms reachable from (resp., that can reach) it in $G_E(\sigma, r)$. for $X \subseteq \mathcal{R}$, let $R(X)$ and $R^{-1}(X)$ denote the set of rooms reachable from, and the set of rooms that can reach, some rooms in $X$, respectively.

\Cref{alg:bounds_on_rents} operates by classifying rooms into four categories based on their current rents: $L$, the set of rooms priced strictly above their lower bounds; $L_f$, those exactly at their lower bounds; $U$, rooms priced strictly below their upper bounds; and $U_f$, those exactly at their upper bounds.

If $L \neq \varnothing$, the algorithm attempts to increase the rents of all rooms in $L$, along with all rooms in their reachable set $R(L)$, at unit rate. To keep the total rent fixed at $R_t$, it simultaneously decreases the rents of rooms not in $R(L_f) \cup R(L)$—since decreasing the rent of a room in $R(L_f)$ would eventually require decreasing the rent of a room in $L_f$, which is already at its lower bound and thus infeasible. 

If no such room exists, then the rent can not be altered without violating the lower bounds. the algorithm returns ``No solution,'' as this constitutes a valid certificate that no envy-free allocation exists within the specified bounds.

The rate of decrease is carefully chosen so that the total rent remains constant. This adjustment continues until one of the following events occurs: 
- A room enters or leaves one of the sets $L$, $U$, $L_f$, or $U_f$, or
- A new weak envy edge is introduced into the graph $G_E(\sigma, r)$. The sets and the rates are recomputed and the process is repeated until $L = \varnothing$.

Once $L = \varnothing$, a symmetric procedure is applied to reduce the rents of rooms in $U$ (i.e., those strictly below their upper bounds), continuing until $U = \varnothing$. As shown in \Cref{thm:bounds_on_rents}, this process produces an envy-free rent vector within bounds, if such a solution exists. \Cref{fig:algo_bounds} gives a picture representation of how the algorithm works.

\begin{restatable}{theorem}{boundsonrents}
$^\dagger$
\label{thm:bounds_on_rents}
\Cref{alg:bounds_on_rents} returns an envy-free rent vector satisfying the given lower and upper bounds, if such a solution exists. Otherwise, it correctly returns "No solution."
\end{restatable}

\subsection{Envy-free rents for maximin and leximin utility among agents}
\label{sec:maximin_and_leximin}

We first describe \Cref{alg:maximin}, which computes an envy-free allocation that maximizes the minimum utility among agents, and later, extend it to obtain leximin utility. The algorithm starts with a rent vector $r$ satisfying envy-freeness and the rent bounds, obtained via \Cref{alg:bounds_on_rents}. It then iteratively increases the minimum utility without violating rent bounds.

Firstly, the algorithm identifies the set $M$ of rooms whose occupants have the minimum utility under $r$. It decreases the rents of rooms in $M$, along with all rooms in $R^{-1}(M)$, at a constant rate to preserve envy-freeness. To maintain the total rent, it increases the rents of rooms that are neither in $R^{-1}(M)$ nor in $R^{-1}(U_f)$—i.e., rooms whose rents can be safely raised while remaining within bounds. The process continues until a new weak envy edge appears or a vertex enters or exits $M$, $L_f$, or $U_f$, at which point the sets and update rates are recomputed. The algorithm terminates when $S_{\text{inc}}$, the set of rooms eligible for rent increase, becomes empty, or when decreasing rents in $M$ would require lowering rents of rooms in $L_f$, violating lower bounds. In \Cref{thm:maximin}, we prove that the final rent vector maximizes the minimum utility subject to envy-freeness and rent bounds. See \Cref{fig:maximin} for details.

\begin{algorithm}[ht]
	\begin{wbox}
    \textsc{Maximin-EF-Rents}$(\sigma,r, \mathcal{A}, \mathcal{R}, V, R_t, l, u)$\\
		\textbf{Input:} A rent division instance with bounds on room rents, $I = (\mathcal{A}, \mathcal{R}, V, R_t, l, u)$, a welfare-maximizing assignment $\sigma$ and an envy-free rent division $r$ satisfying the bounds. \\
        \textbf{Output:} An envy-free rent division $r$ that satisfies the bounds on room rents and maximizes the minimum utility.
		\begin{enumerate}
		  \item \textbf{Initialization:}
		  \begin{enumerate}
                \item $G_E(\sigma, r) := \text{Envy-Graph}(\sigma, r)$
                \item $L_f := \{j\in \calR | r_j = l_j \} $
                \item $U_f := \{j\in \calR | r_j = u_j \} $
                \item $M := \{j\in \mathcal{R}| u_{\sigma^{-1}(j)}=\min_{j'\in \calR} (u_{\sigma^{-1}(j')})\}$
		  \end{enumerate}
            \item \textbf{While} $(R^{-1}(M) \cup R^{-1}(U_f)) \neq \calR $ \textbf{do:}
                \begin{enumerate}
                    \item \textbf{If $ (R(L_f) \cap M) \neq \varnothing$}, \textbf{break};
                    \item Set $S_{dec} := R^{-1}(M)$
                    \item Set $S_{inc} := \calR \setminus (R^{-1}(M) \cup R^{-1}(U_f))$
                    \item Increase rents, $r_j$, of all rooms in $S_{inc}$ at unit rate; \\ Decrease rents, $r_j$, of all rooms in $S_{dec}$ at the rate $\frac{|S_{inc}|}{|S_{dec}|}$ simultaneously;
                    \item Stop when a (weak-envy) edge is added to $G_E(r)$ or a vertex moves into or out of $M,L_f,U_f$.
                    \item Update $M,L_f,U_f \text{ and } G_E(r)$.
            \end{enumerate}
            \item \textbf{Return: $r$} 

		\end{enumerate}
	\end{wbox}
        \caption{Algorithm to find envy-free rents that maximize the minimum utility.}
	\label{alg:maximin} 
\end{algorithm}

\begin{figure}
   \begin{wbox}
       \begin{center}
    
    \begin{tikzpicture}[scale=0.65, >=Stealth, thick, font=\large]

\definecolor{lblue}{RGB}{0,180,255}
\definecolor{lred}{RGB}{255,80,80}
\definecolor{darkgreen}{RGB}{0,130,0}

\begin{scope}[shift={(-5,2)}]
  \draw[lblue, thick] (0,0) ellipse (2 and 2.6);
  \draw[lblue, thick] (0,1) circle (1.2);
  \node at (0,1) {$M$};
  \node at (0,-2) {$R^{-1}(M)$};
  \draw[->, lblue, thick] (-1,-0.8) .. controls (-0.4,0.2) .. (-0.4,0.5);
  \draw[->, lblue, thick] (0,-0.8) .. controls (0,0.2) .. (0,0.5);
  \draw[->, lblue, thick] (1,-1) .. controls (0.4,0.2) .. (0.4,0.5);
  \node at (2.3,1.6) {\small $r \downarrow$};
\end{scope}

 \draw[darkgreen, thick] (1,2) ellipse (3 and 2);
 \node[darkgreen] at (1,2) {\large rest};

\begin{scope}[shift={(-3,-3)}]
  \draw[black, thick] (0,0) ellipse (1.6 and 1.9);
  \draw[black, thick] (0,0.9) circle (0.9);
  \node at (0,1) {$U_f$};
  \node at (0,-1.3) {$R^{-1}(U_f)$};
 \draw[->, lblue, thick] (-0.8,-0.6) .. controls (-0.3,0.2) .. (-0.3,0.5);
  \draw[->, lblue, thick] (0,-0.6) .. controls (0,0.2) .. (0,0.5);
  \draw[->, lblue, thick] (0.8,-0.6) .. controls (0.3,0.2) .. (0.3,0.5);
\end{scope}

\begin{scope}[shift={(2,-3)}]
  \draw[black, thick] (0,0) ellipse (1.6 and 1.9);
  \draw[black, thick] (0,0.9) circle (0.9);
  \node at (0,1) {$L_f$};
  \node at (0,-1.3) {$R^(L_f)$};
   \draw[->, lblue, thick] (-0.3,0.5) .. controls (-0.3,0.2) .. (-0.8,-0.6);
  \draw[->, lblue, thick] (0,0.5) .. controls (0,0.2) .. (0,-0.6);
  \draw[->, lblue, thick] (0.3,0.5) .. controls (0.3,0.2) .. (0.8,-0.6);
 
\end{scope}

\end{tikzpicture}
\end{center}

   \end{wbox}
    \caption{Figure describing \Cref{alg:maximin}}
    \label{fig:maximin}
\end{figure}

\begin{restatable}{theorem}{maximin}
$^{\dagger}$
\label{thm:maximin}
    \Cref{alg:maximin} terminates and returns envy-free rents which maximize the minimum utility among agents. 
\end{restatable}
\Cref{alg:leximin} extends the \Cref{alg:maximin} to compute rents that realize the leximin utility vector. After obtaining a rent vector that maximizes the minimum utility as part of Case 1 above, the algorithm freezes the rents of all rooms in $M \cap R(L_f)$—those with minimum utility whose rents can no longer be decreased—by adding them to both $L_f$ and $U_f$. This effectively freezes their rents. It then identifies the rooms with the next lowest utility among the remaining agents and re-applies the maximin subroutine. This process continues until, like in Case 2 above, $R^{-1}(M) \cup R^{-1}(U_f) = \mathcal{R}$ - the point where no room rent can be changed without making the imputation lexicographically worse. \Cref{thm:leximin} proves that the final rent vector achieves the leximin utility profile. See \Cref{fig:leximin} for details.


\begin{algorithm}[ht]
	\begin{wbox}
    \textsc{Leximin-EF-Rents}$(\sigma,r, \mathcal{A}, \mathcal{R}, V, R_t, l, u)$\\
		\textbf{Input:} A rent division instance with bounds on room rents, $I = (\mathcal{A}, \mathcal{R}, V, R_t, l, u)$, a welfare-maximizing assignment $\sigma$ and an envy-free rent division $r$ satisfying the bounds.\\
        \textbf{Output:} An envy-free rent division $r$ that satisfies the bounds on room rents and makes the utility vector leximin.
		\begin{enumerate}
		  \item \textbf{Initialization:}
		  \begin{enumerate}
                \item $G_E(\sigma, r) := \text{Envy-Graph}(\sigma, r)$.
                \item $L_f := \{j\in \calR | r_j = l_j \} $
                \item $U_f := \{j\in \calR | r_j = u_j \} $
                \item $F := \varnothing $
		  \end{enumerate}
            \item \textbf{While} $(R^{-1}(M) \cup R^{-1}(U_f)) \neq \cal R $ \textbf{do:}
                \begin{enumerate}
                    \item $M = \{ j\in \calR | u_{\sigma^{-1}(j)} = \min_{j'\in ({\calR\setminus F}) } (u_{\sigma(j')})\}$
                    \item \textbf{If $(R(L_f) \cap M) \neq \varnothing$}
                    \begin{enumerate}
                        \item $F := F \cup ((R(L_f) \cap M)$\hfill\emph{(freeze the rents)}
                        \item $L_f:= L_f \cup F $
                        \item $U_f:= U_f \cup F $
                        \item \textbf{continue};
                    \end{enumerate}
                    \item Set $S_{dec} = R^{-1}(M)$
                    \item Set $S_{inc} = \calR \setminus (R^{-1}(M) \cup R^{-1}(U_f))$
                    \item Increase rents, $r_j$, of all rooms in $S_{inc}$ at unit rate; \\ Decrease rents, $r_j$, of all rooms in $S_{dec}$ at the rate $\frac{|S_{inc}|}{|S_{dec}|}$ simultaneously;
                    \item Stop when a (weak-envy) edge is added to $G_E(r)$ or a vertex moves into or out of $M,L_f,U_f$.
                    \item Update $M,L_f,U_f \text{ and } G_E(r)$.
            \end{enumerate}
            \item \textbf{Return: $r$} 

		\end{enumerate}
	\end{wbox}
        \caption{Algorithm to find envy-free rents corresponding to the leximin utility vector.}
	\label{alg:leximin} 
\end{algorithm}
\begin{figure}[ht]
   
    \begin{wbox}

    \begin{center}
    
    \begin{tikzpicture}[scale=0.65, >=Stealth, thick, font=\large]

\definecolor{lblue}{RGB}{0,180,255}
\definecolor{lred}{RGB}{255,80,80}
\definecolor{darkgreen}{RGB}{0,130,0}

\begin{scope}[shift={(-5,2)}]
  \draw[lblue, thick] (0,0) ellipse (2 and 2.6);
  \draw[lblue, thick] (0,1) circle (1.2);
  \node at (0,1) {$M$};
  \node at (0,-2) {$R^{-1}(M)$};
  \draw[->, lblue, thick] (-1,-0.8) .. controls (-0.4,0.2) .. (-0.4,0.5);
  \draw[->, lblue, thick] (0,-0.8) .. controls (0,0.2) .. (0,0.5);
  \draw[->, lblue, thick] (1,-1) .. controls (0.4,0.2) .. (0.4,0.5);
  \node at (2.3,1.6) {\small $r \downarrow$};
\end{scope}

 \draw[darkgreen, thick] (1,2) ellipse (3 and 2);
 \node[darkgreen] at (1,2) {\large rest};

\begin{scope}[shift={(-3,-3)}]
  \draw[black, thick] (0,0) ellipse (1.6 and 1.9);
  \draw[black, thick] (0,0.9) circle (0.9);
  \node at (0,1) {$U_f$};
  \node at (0,-1.3) {$R^{-1}(U_f)$};
 \draw[->, lblue, thick] (-0.8,-0.6) .. controls (-0.3,0.2) .. (-0.3,0.5);
  \draw[->, lblue, thick] (0,-0.6) .. controls (0,0.2) .. (0,0.5);
  \draw[->, lblue, thick] (0.8,-0.6) .. controls (0.3,0.2) .. (0.3,0.5);
\end{scope}

\begin{scope}[shift={(2,-3)}]
  \draw[black, thick] (0,0) ellipse (1.6 and 1.9);
  \draw[black, thick] (0,0.9) circle (0.9);
  \node at (0,1) {$L_f$};
  \node at (0,-1.3) {$R^(L_f)$};
   \draw[->, lblue, thick] (-0.3,0.5) .. controls (-0.3,0.2) .. (-0.8,-0.6);
  \draw[->, lblue, thick] (0,0.5) .. controls (0,0.2) .. (0,-0.6);
  \draw[->, lblue, thick] (0.3,0.5) .. controls (0.3,0.2) .. (0.8,-0.6);

\end{scope}
\begin{scope}[shift={(-7,-3)}]
  \draw[thick] (0,0) ellipse (1.6 and 1.9);
  \node at (0,0) {$F$};
  \end{scope}

\end{tikzpicture}
\end{center}
\end{wbox}

    \caption{Figure describing \Cref{alg:leximin}}
    \label{fig:leximin}
\end{figure}

\begin{restatable}{theorem}{leximin}
$^\dagger$
\label{thm:leximin}
    \Cref{alg:leximin} terminates and returns envy-free rents corresponding to the leximin utility vector. 
\end{restatable}

\subsection{Envy-free rents that minimize the absolute and relative spread of utility among agents}
\label{sec:min_spread}

\begin{algorithm}[ht]
	\begin{wbox}
    \textsc{MinSpread-EF-Rents}$(\sigma,r, \mathcal{A}, \mathcal{R}, V, R_t, l, u)$\\
		\textbf{Input:} A rent division instance with bounds on room rents, $I = (\mathcal{A}, \mathcal{R}, V, R_t, l, u)$, a welfare-maximizing assignment $\sigma$ and an envy-free rent division $r$ satisfying the bounds. \\
        \textbf{Output:} An envy-free rent division $r$ that satisfies the bounds on room rents and minimizes the absolute spread of utilities.
		\begin{enumerate}
		  \item \textbf{Initialization:}
		  \begin{enumerate}
                \item $G_E(\sigma, r) := \text{Envy-Graph}(\sigma, r)$
                \item $L_f := \{j\in \calR | r_j = l_j \} $
                \item $U_f := \{j\in \calR | r_j = u_j \} $
                \item $M := \{j\in \mathcal{R}| u_{\sigma^{-1}(j)}=\min_{j'\in \calR} (u_{\sigma^{-1}(j')})\}$
		  \end{enumerate}
            \item \textbf{Find maximin solution:} \\ 
            \textbf{While} $(R^{-1}(M) \cup R^{-1}(U_f)) \neq \calR $ \textbf{do:}
                \begin{enumerate}
                    \item \textbf{If $ (R(L_f) \cap M) \neq \varnothing$}, \textbf{break};
                    \item Set $S_{dec} := R^{-1}(M)$
                    \item Set $S_{inc} := \calR \setminus (R^{-1}(M) \cup R^{-1}(U_f))$
                    \item Increase rents, $r_j$, of all rooms in $S_{inc}$ at unit rate; \\ Decrease rents, $r_j$, of all rooms in $S_{dec}$ at the rate $\frac{|S_{inc}|}{|S_{dec}|}$ simultaneously;
                    \item Stop when a (weak-envy) edge is added to $G_E(r)$ or a vertex moves into or out of $M,L_f,U_f$.
                    \item Update $M,L_f,U_f \text{ and } G_E(r)$.
		  \end{enumerate}
            \item \textbf{Fix agents with maximin utilities:} 
            \begin{enumerate}
                \item $F := M$
                \item $L_f:= L_f \cup F $
                \item $U_f:= U_f \cup F $
                \item $M := \{j\in \mathcal{R}| u_{\sigma^{-1}(j)}=\min_{j'\in (\calR \setminus F)} (u_{\sigma^{-1}(j')})\}$
            \end{enumerate}
                
            \item \textbf{Find minimax solution \emph{among the rest}:} \\
            \textbf{While} $(R(L_f) \cup R(M)) \neq \calR $ \textbf{do:}
                \begin{enumerate}
                    \item \textbf{If $(R^{-1}(U_f) \cap M) \neq \varnothing$}, \textbf{break};
                    \item Set $S_{dec} = \calR \setminus (R(M) \cup R(L_f))$
                    \item Set $S_{inc} = R(M)$
                    \item Increase rents, $r_j$, of all rooms in $S_{inc}$ at unit rate; \\ Decrease rents, $r_j$, of all rooms in $S_{dec}$ at the rate $\frac{|S_{inc}|}{|S_{dec}|}$ simultaneously;
                    \item Stop when a (weak-envy) edge is added to $G_E(r)$ or a vertex moves into or out of $M,L_f,U_f$.
                    \item Update $M,L_f,U_f \text{ and } G_E(r)$.
            \end{enumerate}
            \item \textbf{Return: $r$} 

		\end{enumerate}
	\end{wbox}
        \caption{Algorithm to find envy-free rents corresponding to the leximin utility vector.}
	\label{alg:min_spread} 
\end{algorithm}
\begin{figure}[ht]
   \begin{wbox}
       \begin{center}
    
    \begin{tikzpicture}[scale=0.65, >=Stealth, thick, font=\large]

\definecolor{lblue}{RGB}{0,180,255}
\definecolor{lred}{RGB}{255,80,80}
\definecolor{darkgreen}{RGB}{0,130,0}

\begin{scope}[shift={(-5,2)}]
  \draw[lred, thick] (0,0) ellipse (2 and 2.6);
  \draw[lred, thick] (0,1) circle (1.2);
  \node at (0,1) {$M$};
  \node at (0,-2) {$R(M)$};
  \draw[->, lblue,, thick] (-0.4,0.5) .. controls (-0.4,0.2) .. (-1,-0.8);
  \draw[->, lblue,, thick] (0,0.5) .. controls (0,0.2) .. (0,-0.8);
  \draw[->, lblue,, thick] (0.4,0.5) .. controls (0.4,0.2) .. (1,-1);
  \node at (-3,0.5) {\small $r \uparrow$};
\end{scope}

 \draw[darkgreen, thick] (1,2) ellipse (3 and 2);
 \node[darkgreen] at (1,2) {\large rest};

\begin{scope}[shift={(-2,-3)}]
  \draw[black, thick] (0,0) ellipse (1.6 and 1.9);
  \draw[black, thick] (0,0.9) circle (0.9);
  \node at (0,1) {$U_f$};
  \node at (0,-1.3) {$R^{-1}(U_f)$};
 \draw[->, lblue, thick] (-0.8,-0.6) .. controls (-0.3,0.2) .. (-0.3,0.5);
  \draw[->, lblue, thick] (0,-0.6) .. controls (0,0.2) .. (0,0.5);
  \draw[->, lblue, thick] (0.8,-0.6) .. controls (0.3,0.2) .. (0.3,0.5);
\end{scope}

\begin{scope}[shift={(2,-3)}]
  \draw[black, thick] (0,0) ellipse (1.6 and 1.9);
  \draw[black, thick] (0,0.9) circle (0.9);
  \node at (0,1) {$L_f$};
  \node at (0,-1.3) {$R^(L_f)$};
   \draw[->, lblue, thick] (-0.3,0.5) .. controls (-0.3,0.2) .. (-0.8,-0.6);
  \draw[->, lblue, thick] (0,0.5) .. controls (0,0.2) .. (0,-0.6);
  \draw[->, lblue, thick] (0.3,0.5) .. controls (0.3,0.2) .. (0.8,-0.6);
 
\end{scope}

\begin{scope}[shift={(-6,-3)}]
  \draw[thick] (0,0) ellipse (1.6 and 1.9);
  \node at (0,0) {$F$};
 
\end{scope}

\end{tikzpicture}
\end{center}

   \end{wbox}
    \caption{Figure describing \Cref{alg:min_spread}}
    \label{fig:min-spread}
\end{figure}
\Cref{alg:min_spread} computes the rents that minimize the absolute spread in utilities. It first maximizes the minimum utility using \Cref{alg:maximin} as a subroutine. Then, it freezes the rents of \emph{these} rooms and minimizes the maximum utility among the \emph{remaining} agents using an algorithm analogous to \Cref{alg:maximin}. As \Cref{thm:min_absolute_spread} shows, this procedure minimizes the absolute spread in utilities. Furthermore, \Cref{thm:min_relative_spread} shows that the rent vector minimizing the absolute spread also minimizes the relative spread. Hence, \Cref{alg:min_spread} simultaneously minimizes both absolute and relative utility spread.

\begin{restatable}{theorem}{minabsolutespread}
$^\dagger$
\label{thm:min_absolute_spread}
    \Cref{alg:min_spread} returns envy-free rents that minimizes the absolute spread of utility among the agents. 
\end{restatable}

\begin{restatable}{theorem}{minrelativespread}
$^\dagger$
\label{thm:min_relative_spread}
    \Cref{alg:min_spread} returns envy-free rents that also minimizes the relative spread of utility among the agents. 
\end{restatable}

\begin{remark}
    Note that the algorithms discussed in this section change the rents at constant rate and ensure that the algorithms terminate. It is easy to compute when the next events - vertices movement to different sets or appearance of new weak-envy edges - happens. And so, we can adjust these rents discreetly to the next event. This way, all algorithms can be made to run in strongly polynomial time. 
\end{remark}

\section{Envy-Free Allocation with Room-Specific Agent Budgets}
\label{sec:agent_budgets}

In this section, we consider settings with only room-specific budgets for agents. There are no bounds on the room rents.
Let $I = (\mathcal{A}, \mathcal{R}, V, R_t, b)$ be an instance of the rent division problem where $b = (b_{ij})_{(i, j) \in \mathcal{A} \times \mathcal{R}}$ specifies agent-specific budget constraints for each room. Define $EF(I)$ as the set of all envy-free allocations and $EF_b(I)$ as those that also respect the budget constraints.
For a given envy-free allocation $(\sigma, r)$, let $G_
E(\sigma, r)$ denotes its corresponding envy graph. Since $(\sigma, r)$ is envy-free, all edges in $G_E(\sigma, r)$ must be weak envy edges.

The structure of these envy graphs reveals several important properties. Notably, as established in \cite{AES14components}, the strongly connected components of the envy graph are invariant across all allocations in $EF(I)$:

\begin{lemma}[\cite{AES14components}]
\label{lem:conn_components}
The strongly connected components of $G_E(\sigma, r)$ remain unchanged for all $(\sigma, r) \in EF(I)$.
\end{lemma}

This observation implies that the strongly connected components (SCCs) of the envy graph depend only on the instance $(\mathcal{A}, \mathcal{R}, V, R_t)$, and is invariant across all envy-free allocations. In fact, the total rent $R_t$ has no bearing on the SCC decomposition: envy-free rent vectors can be adjusted to accommodate different total rents without changing the underlying envy relations. Specifically, if $(\sigma, r) \in EF(\mathcal{A}, \mathcal{R}, V, R_t)$ is an envy-free allocation, then for any alternative total rent $R_t'$, we can construct a new rent vector $r'$ defined by $r'_j = r_j + \frac{R_t' - R_t}{|\mathcal{A}|}$ for all $j \in \mathcal{R}$ such that $(\sigma, r')$ is also envy-free. Since the envy graph remains unchanged, the set of SCCs is preserved.

We denote this canonical partition by $\mathcal{C}(\mathcal{A}, \mathcal{R}, V)$, which splits the room set $\mathcal{R}$ into SCCs shared across all envy-free allocations. As argued in \cite{Procaccia_Velez_Yu_2018}, instances in which the envy graph consists of a single SCC are particularly well-structured and exhibit the following properties.

\begin{lemma}[\cite{Procaccia_Velez_Yu_2018}]
\label{lem:single_SCC}
Let $I = (\mathcal{A}, \mathcal{R}, V, R_t)$ be an instance such that $\mathcal{C}(\mathcal{A}, \mathcal{R}, V) = \{\mathcal{R}\}$. Then, for all $R_t \in \mathbb{R}$:
\begin{enumerate}
    \item If $\{(\sigma, r), (\mu, r')\} \subseteq EF(\mathcal{A}, \mathcal{R}, V, R_t)$, then $r = r'$.
    \item Each agent has the same utility among all allocations in $EF(\mathcal{A}, \mathcal{R}, V, R_t)$.
    \item Let $(\sigma, r) \in EF(\mathcal{A}, \mathcal{R}, V, R_t)$. Then for any $R'_t \in \mathbb{R}$, and $(\mu, r') \in EF(\mathcal{A}, \mathcal{R}, V, R'_t)$, it holds that for all $a \in \mathcal{R}$, $r'_a = r_a + (R'_t - R_t)/|\mathcal{R}|$.
    \item The rent of each room in any allocation in $EF(\mathcal{A}, \mathcal{R}, V, R_t)$ is an increasing function of $R_t$.
\end{enumerate}
\end{lemma}

In fact, as shown below, welfare-maximizing solutions also have nice properties corresponding to the structure of the strongly connected components.

\begin{lemma}[\cite{Procaccia_Velez_Yu_2018}]
\label{lem:matchings_and_conn_components}
Let $I = (\mathcal{A}, \mathcal{R}, V, R_t)$ be a rent division instance. Then, for any two welfare-maximizing assignments $\sigma$ and $\mu$, the allocation of rooms to agents within each component $C \in \mathcal{C}(\mathcal{A}, \mathcal{R}, V)$ remains the same, i.e., $\sigma(C) = \mu(C)$.
\end{lemma}

Consequently, in every welfare-optimal assignment, the rooms in any strongly connected component $C$ are allocated exclusively among a fixed subset of agents. This property allows us to decompose the global problem into smaller subproblems—one for each strongly connected component—and solve them independently. The resulting component-wise solutions can then be combined to form a solution for the original instance.

Consider a strongly connected component $C \in \mathcal{C}(\mathcal{A}, \mathcal{R}, V)$. Let $\mathcal{A}(C)$ be the set of agents who are assigned rooms in $C$ under some envy-free allocation. Define the restricted valuation and budget functions as $v_C = (v_{ij})_{i \in \mathcal{A}(C), j \in C}$ and $b_C = (b_{ij})_{i \in \mathcal{A}(C), j \in C}$. Let $\delta_C$ denote the maximum total rent with individual rationality for which an envy-free allocation respecting the budgets exists in this subproblem:
$$
EF_{b_C}(\mathcal{A}(C), C, v_C, \delta_C) \neq \emptyset. 
$$
Continuity of utility functions guarantees the existence of such a threshold $\delta_C$.

\begin{lemma}
\label{lem:join_matchings}
Let $I = (\mathcal{A}, \mathcal{R}, V, R_t)$ be a rent division instance with agent-room-specific budgets $b = (b_{ij})_{i,j}$. Suppose that for each strongly connected component $C \in \mathcal{C}(\mathcal{A}, \mathcal{R}, V)$, there exists a valid envy-free allocation $(\sigma_C, r_C) \in EF_{b_C}(\mathcal{A}(C), C, v_C, \delta_C)$. Construct a global assignment $\sigma$ by setting $\sigma(i) = \sigma_C(i)$ for each agent $i \in \mathcal{A}(C)$. Then, for any envy-free allocation $(\sigma', r') \in EF_b(I)$, it holds that $(\sigma, r') \in EF_b(I)$ as well.
\end{lemma}

\begin{proof}
Assume that $(\sigma', r') \in EF_b(I)$. We aim to show that $(\sigma, r') \in EF_b(I)$.

First, observe that $\sigma$ is a welfare-maximizing matching, and $(\sigma', r') \in EF(I)$. By \cref{thm:welfare}, it follows that $(\sigma, r') \in EF(I)$ as well. It remains to verify that $(\sigma, r')$ satisfies the budget constraints.

We now prove this property for each strongly connected component (SCC) $C \in \mathcal{C}(\mathcal{A}, \mathcal{R}, V)$ separately. Let $\mathcal{A}(C)$ denote the set of agents assigned to rooms in $C$ under $\sigma$, and define:
$$
\delta'_C := \sum_{i \in \mathcal{A}(C)} r'_{\sigma'(i)}.
$$
By \cref{lem:single_SCC}, the unique rent vector $r_C$ corresponding to $\sigma_C$ (the restriction of $\sigma$ to $C$) satisfies:
\begin{math}
r_C = r'_C + (\delta_C - \delta'_C)\mathbf{1},
\end{math}
where, $\delta_C$ denotes the maximum total rent for component $C$ such that a solution exists in $EF_{b_C}(\mathcal{A}(C), C, v_C, \delta_C)$. Since $\delta'_C \leq \delta_C$, it follows that
$$
r'_C = r_C - (\delta_C - \delta'_C)\mathbf{1} \preceq r_C.
$$
Because $(\sigma_C, r_C)$ satisfies the upper budget bounds by construction, and $r'_C$ is obtained by uniformly decreasing each coordinate of $r_C$, the upper bounds remain satisfied. Therefore, $(\sigma_C, r'_C)$ also satisfies the budget constraints. Hence, $(\sigma, r') \in EF_b(I)$.
\end{proof}

\textbf{Budget–aware envy graph.} For a profile $(\sigma,r)$ and upper–bound vector $b=(b_{ij})_{(i,j)\in\mathcal A\times\mathcal R}$, define
$$
  G_b(\sigma,r) \;=\; (\mathcal R,E), 
  \quad
  E \;=\;
  \bigl\{\,
      (\sigma(i),\sigma(j))
      \;\bigm|\;
      u_{i\sigma(i)} \le u_{i\sigma(j)} 
      \text{ and } 
      r_{\sigma(j)} < b_{i\sigma(j)}
  \bigr\}.
$$

Note that, an edge from room $i$ to room $j$ implies that the agent in room $i$, can move to room $j$ with neither decreasing their utility nor paying over their budget - representing a feasible change of rooms.  

\begin{algorithm}[ht]
  \begin{wbox}
  \textsc{SCC-EF-Rents}$(\mathcal{A}, \mathcal{R}, V, b)$\\
    \textbf{Input:} A rent–division instance
      $(\mathcal A,\mathcal R,V,b)$ with upper bounds
      $b=(b_{ij})_{i\in\mathcal A,j\in\mathcal R}$.\\
    \textbf{Output:} The maximum feasible total rent $R_t$ and an
      envy-free allocation $(\sigma,r)$ that attains it.
    \begin{enumerate}
      \item \textbf{Initial EF solution:}\;
            Compute $(\sigma,r)\in EF(\mathcal A,\mathcal R,V,R_t)$ for
            some $R_t\in\mathbb R$.
      \item Choose $\Delta\in\mathbb R$ such that
            \[
              \bigl(\sigma,(r_a-\Delta)_{a\in\mathcal R}\bigr)\in
              EF_b(\mathcal A,\mathcal R,V,R_t-n\Delta)
              \quad\text{and}\quad
              \exists\,i\in\mathcal A:\;
               r_{\sigma(i)} = b_{i\sigma(i)}.
            \]
      \item $r \leftarrow (r_{\sigma(i)}-\Delta)_{i\in\mathcal A}$.
      \item $R_t \leftarrow R_t-n\Delta$.
      \item \textbf{while} there is \emph{no} $i\in\mathcal A$ with
            $r_{\sigma(i)} = b_{i\sigma(i)}$
            that is \emph{not} on a cycle of $G_b(\sigma,r)$ \textbf{do}
            \begin{enumerate}
              \item \textbf{if} There is no agent $i$ such that $r_i = b$ 
                    \textbf{then}\hfill\emph{(Case 1)}
                    \begin{enumerate}
                      \item $\displaystyle
                              \Delta \leftarrow
                              \min_{i\in\mathcal A}
                              \bigl(b_{i\sigma(i)}-r_{\sigma(i)}\bigr)$
                      \item $r \leftarrow (r_a+\Delta)_{a\in\mathcal R}$
                      \item $R_t \leftarrow R_t+n\Delta$
                    \end{enumerate}
              \item \textbf{else}\hfill\emph{(Case 2)}
                    \begin{enumerate}
                      \item Find $i\in\mathcal A$ with
                            $r_{\sigma(i)} = b_{i\sigma(i)}$
                            that lies on a cycle
                            $C$ of $G_b(\sigma,r)$.
                      \item $\sigma \leftarrow$ reshuffle of $\sigma$ along $C$
                    \end{enumerate}
            \end{enumerate}
      \item \textbf{end while}
      \item \textbf{Return} $R_t$ and $(\sigma,r)$.
    \end{enumerate}
  \end{wbox}
  \caption{Maximum-rent envy-free allocation inside one strongly-connected
           component; $G_b(\sigma,r)$ is the budget-aware envy graph.}
  \label{alg:max_rent_EF_budgets}
\end{algorithm}

\begin{lemma}\label{lem:best_matching}
Let $(\mathcal A,\mathcal R,V)$ be such that the entire instance forms
a single strongly-connected component,
$\mathcal C(\mathcal A,\mathcal R,V)=\{\mathcal R\}$, and let
$b=(b_{ij})_{i,j}$ be the upper-bound matrix.
Then \Cref{alg:max_rent_EF_budgets} runs in polynomial time.  If
$(\sigma,r)$ is its output with total rent $R_t$, then
$(\sigma,r)\in EF_b(\mathcal A,\mathcal R,V,R_t)$ and the value
$\delta_{\mathcal A}=R_
t$ is the maximum: for every $R'_t>R_t$,
$EF_b(\mathcal A,\mathcal R,V,R'_t)=\emptyset$.
\end{lemma}

\begin{proof}
First, we prove that the algorithm terminates in polynomial time. The initial envy-free (EF) solution can be computed in polynomial time using the algorithm of \cite{klijn2000}. The main loop in line~5 has two cases. Case 1 can occur at most $n^2$ times, since in each iteration, the rent of a room is increased to match the budget of an agent for that room. Because rents are monotonically increasing, this adjustment cannot occur again for the same agent-room pair. As there are at most $n^2$ such pairs, Case 1 is executed at most $n^2$ times. Case 2 can occur at most $n$ times consecutively after which, either Case 1 must be triggered or the algorithm must terminate. Therefore, the total number of iterations is bounded by a polynomial in $n$.\\
Next, we claim that for each $R'_t > R_t$, $\mathsf{EF}_b(\mathcal{A}, \mathcal{R}, V, R'_t) = \emptyset$. Suppose for the sake of contradiction that there exists $R'_t > R_t$ and $(\mu, r') \in \mathsf{EF}_b(\mathcal{A}, \mathcal{R}, V, R'_t)$. Let $\Delta \coloneqq (R'
_t- R_t)/n$. By \cref{lem:single_SCC}, for each $a \in \mathcal{R}$, we have $r'_a = r_a + \Delta$. Recall that $(\mu, r) \in \mathsf{EF}_b(\mathcal{A}, \mathcal{R}, V, R_t)$. Due to \cref{cor:equal_utility}, we have that for all $i \in \mathcal{A}$,
$$
v_{i\sigma(i)} - r_{\sigma(i)} = v_{i\mu(i)} - r_{\mu(i)}.
$$
Along with
$$
r_{\mu(i)} = r'_{\mu(i)} - \Delta < b_i,
$$
it follows that $(\sigma(i), (\mu(i)))$ must be an edge of $G_b(\sigma, r)$ for any $i$. Let $\nu = \sigma^{-1} \circ \mu$, then for all $i \in \mathcal{A}$, we have
$$
\sigma(i) \rightarrow \nu(\sigma(i)) \rightarrow \nu^2(\sigma(i)) \rightarrow \cdots \rightarrow \sigma(i)
$$
is a cycle in $G_b(\sigma, r)$. Also, notice that the while loop only terminates when there is some $i \in \mathcal{A}$ such that this condition fails, contradicting the assumption that $r_{\sigma(i)} = b_i$. Therefore, agent $i$ satisfies the conditions of Case 2, and the algorithm should not have terminated.
\end{proof}

\medskip
\noindent\textbf{Combining components.}\;
Apply \Cref{alg:max_rent_EF_budgets} independently to every SCC
of $\mathcal C(\mathcal A,\mathcal R,V)$ and merge the results as
described below.\\
\begin{algorithm}[ht]
  \begin{wbox}
    \textsc{BudgetAware-EF}$(\mathcal{A}, \mathcal{R}, V, R_t, b)$ \\[0.5em]
    \textbf{Input:} A rent division instance $(\mathcal{A}, \mathcal{R}, V, R_t, b)$ with target total rent $R_t$ and agent-specific budgets $b = (b_{ij})$. \\
    \textbf{Output:} An envy-free allocation $(\mu, r)$ satisfying the budgets and summing to $R_t$, or ``No solution''.
    \begin{enumerate}
        \item Compute an initial EF allocation $(\sigma, r) \in EF(\mathcal{A}, \mathcal{R}, V, R')$ for some $R'$.
        \item Let $\mathcal{C} \gets$ strongly connected components of $G_E(\sigma, r)$.
        \item \textbf{For each} component $C \in \mathcal{C}$:
        \begin{enumerate}
            \item $(\mu_C, r_C) \gets$ \textsc{SCC-EF-Rents}$(\mathcal{A_C}, \mathcal{R_C}, V, b)$\hfill\Cref{alg:max_rent_EF_budgets}
        \end{enumerate}
        \item Combine $\mu_C$ into a global assignment $\mu$
        \item  Compute rent division $r$ such that $(\mu,r) \in EF(\mathcal{A},\mathcal{R},V,R_t)$.
        \item $r$ $\gets$ \textsc{BoundsAware-EF-Rents}$(\mu,r, \mathcal{A}, \mathcal{R}, V, R_t, -\infty, b_{i\mu(i)})$\hfill\Cref{alg:bounds_on_rents}
        \item \textbf{Return} $(\mu, r)$.
    \end{enumerate}
  \end{wbox}
  \caption{Computing EF allocation under room-specific budgets for agents.}
  \label{alg:global_budget}
\end{algorithm}

\begin{theorem}\label{thm:budget_algo}
\Cref{alg:global_budget} runs in polynomial time and returns an allocation
in $EF_b(\mathcal A,\mathcal R,V,R_t)$ whenever such an allocation exists;
otherwise it outputs ``no solution''.
\end{theorem}

\begin{proof}
The algorithm processes each strongly connected component (SCC) by first running \cref{alg:max_rent_EF_budgets}, followed by a single execution of \cref{alg:bounds_on_rents}. Since both subroutines run in polynomial time, the overall algorithm is polynomial-time. The final output is the result of \cref{alg:bounds_on_rents}, and is therefore envy-free; Moreover, by construction of the instance, the solution respects the budget constraints: for each agent $i$ assigned to room $\mu(i)$, we impose the upper bound $b_{i\mu(i)}$, ensuring that the resulting rent respects $i$'s budget. By \cref{lem:join_matchings} and \cref{thm:welfare_bounds}, if a feasible solution exists, it must exist for the assignment $\mu$. Thus, if the algorithm returns ``no solution'' in line 5, no feasible EF allocation exists for any other assignment.
\end{proof}

\begin{corollary}
    By substituting \cref{alg:bounds_on_rents} subroutine in line 5 of \cref{alg:global_budget} with \cref{alg:maximin}, \cref{alg:leximin}, or \cref{alg:min_spread}, the algorithm computes the maximin, leximin, or minimum-spread envy-free allocations, respectively, under room-specific agent budget constraints.
\end{corollary}



\section{Envy-free allocation with bounds on room rents and room-specific agent budgets}
\label{sec:rent_bounds_budgets}

In this section, we investigate the most general setting of the rent division problem, where we are given both room-specific lower and upper bounds on rents and agent-specific budgets for rooms. Formally, the input instance is $I = (\mathcal{A}, \mathcal{R}, V, R_t, b, l, u)$, where $b = (b_{ij})$ specifies the maximum rent agent $i$ is willing to pay for room $j$, and $l = (l_j)$ and $u = (u_j)$ define lower and upper bounds on each room’s rent.

\begin{theorem}
\label{thm:combined}
There exists a polynomial-time algorithm that computes an envy-free allocation $(\sigma, r)$ satisfying both the budget constraints $b$ and the room rent bounds $l, u$. Furthermore, the algorithm can be modified to return an envy-free allocation optimizing any of the following objectives among all feasible envy-free allocations:
    (i) lexicographically maximal utilities (leximin),
    (ii) lexicographically minimal utilities (leximax),
    (iii) maximum among minimum utilities (maximin),
    (iv) minimum utility spread (min-spread).
\end{theorem}

\begin{algorithm}[ht]
\caption{Combined EF allocation with budgets and rent bounds}
\label{alg:combined}
\begin{wbox}
\textsc{BudgetAndBoundsAware-EF-Rents}$(\mathcal{A}, \mathcal{R}, V, R_t, b,  l, u)$ \\
\textbf{Input:} Instance $I = (\mathcal{A}, \mathcal{R}, V, R_t, b, l, u)$\\
\textbf{Output:} Envy-free allocation $(\sigma, r)$ satisfying all constraints
\begin{enumerate}
    \item $\sigma \gets$ \textsc{BudgetAware-EF}$(\mathcal{A}, \mathcal{R}, V, R_t, b)$\hfill(\cref{alg:global_budget})
    \item \textbf{for each} room $j \in \mathcal{R}$:
    \begin{itemize}
        \item Let $i = \sigma^{-1}(j)$ be the agent assigned to room $j$
        \item Set $\tilde{u}_j \gets \min\{ u_j,\ b_{ij} \}$
    \end{itemize}
    \item $r \gets$ \textsc{BoundsAware-EF-Rents}$(\sigma, \mathcal{A}, \mathcal{R}, V, R_t, l, \tilde{u})$\hfill(\cref{alg:bounds_on_rents})
    \item \textbf{Return} $(\sigma, r)$
\end{enumerate}
\end{wbox}
\end{algorithm}

\begin{proof}
Run \cref{alg:global_budget} on $(\mathcal{A,R},V,b,R_t)$ to computes an envy-free allocation $(\sigma, r')$ that satisfies the agent-specific budget constraints $b$ but may violate the room rent bounds $l$ and $u$. 

To ensure that the final rents also respect the room bounds, we define new upper bounds $\tilde{u}_j := \min\{u_j, b_{\sigma^{-1}(j)\, j}\}$ for each room $j$. These bounds guarantee that the rent paid for room $j$ will not exceed either the agent's budget for room $j$ or the specified room cap.

Next, we use \cref{alg:bounds_on_rents} to compute a new rent vector $r$ satisfying the bounds $l_j \le r_j \le \tilde{u}_j$, while preserving envy-freeness with respect to the fixed assignment $\sigma$. This ensures that the final allocation $(\sigma, r)$ satisfies both types of constraints.

Since each step uses polynomial-time subroutines and the number of rooms is $n$, the overall algorithm runs in polynomial time.

For the optimization objectives: Once the assignment $\sigma$ is fixed, the problem becomes finding an envy-free solution with a fixed assignment and bounds on rents $l_j \le r_j \le \tilde u_j$, that in \Cref{sec:rent_bounds} we have provided solutions for different problems in this setting such as: \emph{maximin} can be solved via \Cref{alg:maximin}, \emph{leximin} can be solved via \Cref{alg:leximin}, and \emph{min-spread} can be solved via \Cref{alg:min_spread}.\\
Note that if there exists an envy-free allocation $(\mu, r')$ that satisfies both types of constraints, then by \cref{lem:best_matching} and \cref{thm:welfare_bounds}, the allocation $(\sigma, r')$ must also satisfy both types of constraints. Therefore, if the algorithm outputs ``no solution'' for the assignment $\sigma$, it implies that no feasible envy-free solution exists for any assignment.
\end{proof}

\begin{remark}
   Given a fixed assignment of agents to rooms, all envy-free, rent bounds, and budget constraints can be written as linear constraints. Consequently, once we have found the matching, the fair envy-free rents can be computed under that assignment and can also be computed via linear programming, refer to Procaccia et al.~\cite{Procaccia_Velez_Yu_2018} for more details.
\end{remark}

\begin{example}[Leximin vs. Min-Spread]
Consider the following rent division instance with four agents and four rooms. The valuation matrix $V \in \mathbb{R}^{4 \times 4}_+$  and its unique maximum weight matching is shown below:

$$
V =
\begin{bmatrix}
\highlight{20} & 0 & 20 & 0 \\
0 & \highlight{19} & 0 & 0 \\
5 & 0 & \highlight{5} & 0 \\
0 & 0 & 0 & \highlight{2}
\end{bmatrix}
$$

The total rent is $4$ and each room $j$ must be priced within the bounds $[l_j, u_j]$ shown below:

$$
[l_1, u_1] = [0,2], \quad [l_2, u_2] = [0,2], \quad [l_3, u_3] = [0,2], \quad [l_4, u_4] = [2,2]
$$

Two rent vectors (EF under the same assignment) are:

$$
\text{Leximin:} \quad r = (0, 2, 0, 2) \quad \Rightarrow \quad u = (20, 8, 5, 0)
$$
$$
\text{Min-Spread:} \quad r = (1, 0, 1, 2) \quad \Rightarrow \quad u = (19, 10, 4, 0)
$$

Both are envy-free and satisfy the rent bounds, but they differ in their fairness objectives:
\begin{itemize}
    \item The first maximizes the lexicographically minimum utility vector (\textbf{leximin}).
    \item The second minimizes the difference between maximum and minimum utilities (\textbf{min-spread}).
\end{itemize}

This illustrates that the leximin solution does not always yield the minimum spread.
\end{example}




\section*{Acknowledgments}
This work was supported in part by the National Science Foundation Grant CCF-2230414.

\bibliographystyle{plainurl}
\bibliography{refs}

\appendix

\section{Appendix}

\subsection{Klijn's Algorithm}
\label{subsec:Klijn}
Given an instance of the rent division problem $I = (\calA,\calR,V,R_t)$, begin with an arbitrary allocation of rooms and rents, $(\sigma,r)$, and construct the corresponding envy graph $G_E(\sigma, r)$. If the graph contains a cycle with at least one strong envy edge, then reassign(/permute) agents along the cycle so that each agent moves to the room they weakly prefer within the cycle.

If there exists a strong envy edge $(i,j)$ that is not part of any cycle, proceed as follows. Let $R^{-1}(i)$ denote the set of rooms from which there is a directed path to room $i$ (the inverse reachability set), and let $R(j)$ denote the set of rooms reachable from room $j$ (the reachability set). Since $(i,j)$ is not part of any cycle, these two sets are disjoint.

Now, increase the rents of all rooms in $R^{-1}(i)$ at unit rate, and decrease the rents of all rooms in $R(j)$ at rate $\frac{|R^{-1}(i)|}{|R(j)|}$. This adjustment preserves the total rent. As rents evolve under this scheme, one of two things will happen: either the strong edge $(i,j)$ becomes weak, or new weak edges appear that create a cycle containing at least one strong edge or increase the size of $R^{-1}(i) \cup R(j)$. 

In either case, Klijn~\cite{klijn2000} proves that this iterative procedure makes progress by an elimination of a strong edge and hence the algorithm terminates in strongly polynomial time. For a full description and analysis, see~\cite{klijn2000}.

\begin{figure}
    \centering

    \begin{wbox}
        \begin{center}
    \begin{tikzpicture}[scale=0.65, >=Stealth, thick, font=\large]

\definecolor{lblue}{RGB}{0,180,255}
\definecolor{lred}{RGB}{255,80,80}
\definecolor{darkgreen}{RGB}{0,130,0}
\begin{scope}[shift={(-5,2)}]
 \draw[thick] (0,1) circle (0.5);
  \node at (0,1) {$i$};
\end{scope}
\begin{scope}[shift={(5,2)}]
    \draw[ thick] (0,1) circle (0.5);
  \node at (0,1) {$j$};
  \end{scope}
  \draw[->, thick, red] (-4.5,2+1) -- (4.5+0,2+1); 
  
\begin{scope}[shift={(-5,2)}]
  \draw[ thick] (0,0) ellipse (2 and 2.6);
 
  \node at (0,-2) {$R^{-1}(i)$};
  \draw[->, lblue, thick] (-1,-0.8) .. controls (-0.4,0.2) .. (-0.4,0.5);
  \draw[->, lblue, thick] (0,-0.8) .. controls (0,0.2) .. (0,0.5);
  \draw[->, lblue, thick] (1,-1) .. controls (0.4,0.2) .. (0.4,0.5);
  \node at (-2.5,0) { $r \downarrow$};
\end{scope}

\begin{scope}[shift={(5,2)}]
  \draw[thick] (0,0) ellipse (2 and 2.6);

  \node at (0,-2) {$R(j)$};
  \draw[->, lblue, thick] (-0.4,0.5) .. controls (-0.4,0.2) .. (-1,-0.8);
  \draw[->, lblue, thick] (0,0.5) .. controls (0,0.2) .. (0,-0.8);
  \draw[->, lblue, thick] (0.4,0.5) .. controls (0.4,0.2) .. (1,-1);
  \node at (2.5, 0) { $r \uparrow$};
\end{scope}
 \draw[darkgreen, thick] (0,-4) ellipse (2.5 and 1.8);
 \node[darkgreen] at (0,-4) {\large rest};

\end{tikzpicture}
\end{center}

    \end{wbox}
    
    \caption{Figure describing algorithm in \Cref{subsec:Klijn}}
\end{figure}
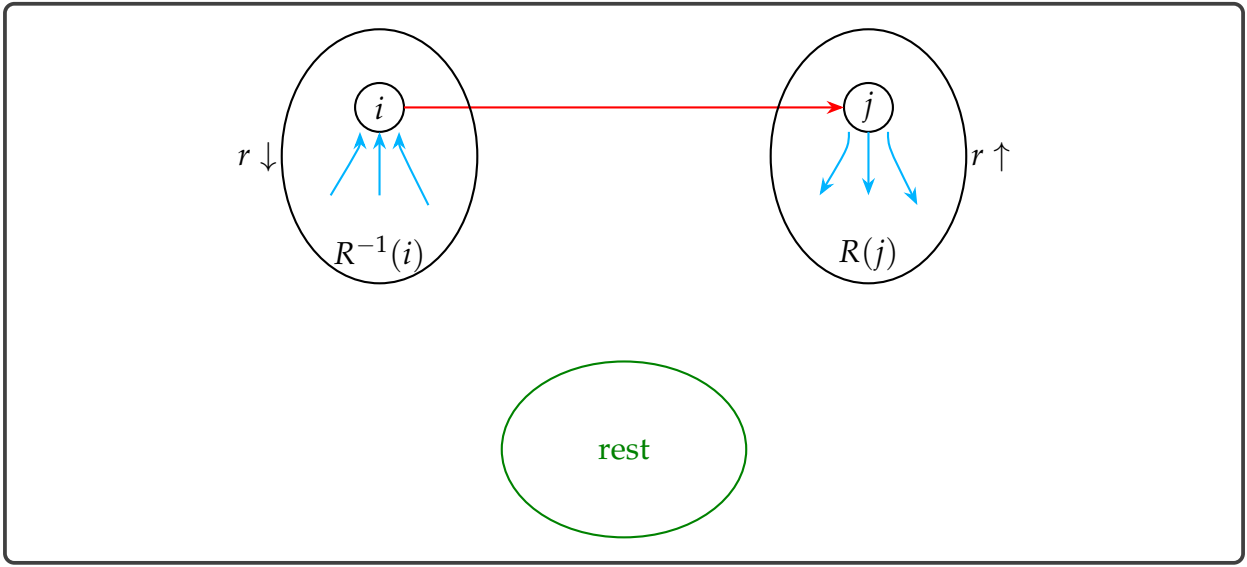

\begin{observation}
For any rent vector $r$, the envy graph corresponding to a welfare-maximizing assignment contains no cycles with strong envy edges. Consequently, if Klijn's algorithm \cite{klijn2000} is initialized with a welfare-maximizing assignment, the matching remains fixed throughout, and only the rents need to be adjusted to eliminate all strong envy edges and achieve envy-freeness.
\end{observation}

\section{Omitted Proofs}
\label{app:proofs}
\welfarebounds*
\begin{proof}
\label{proof:welfare_bounds}
Let $\sigma'$ be any welfare-maximizing assignment. Since $(\sigma, r)$ is envy-free, each agent $i$ weakly prefers their allocation under $\sigma$ to that under $\sigma'$:
$$
v_{i\sigma(i)} - r_{\sigma(i)} \geq v_{i\sigma'(i)} - r_{\sigma'(i)} \quad \text{for all } i \in \mathcal{A}.
$$
Summing over all agents and rearranging,
$$
\sum_{i \in \mathcal{A}} v_{i\sigma(i)} - \sum_{j \in \mathcal{R}} r_j \geq \sum_{i \in \mathcal{A}} v_{i\sigma'(i)} - \sum_{j \in \mathcal{R}} r_j.
$$
Since the total rent $\sum_j r_j = R_t$ is fixed, we get:
$
\sum_{i \in \mathcal{A}} v_{i\sigma(i)} \geq \sum_{i \in \mathcal{A}} v_{i\sigma'(i)}.
$
But $\sigma'$ is welfare-maximizing, so this inequality must be tight. Hence, $\sigma$ is also welfare-maximizing, and the initial inequalities for each agent must also be equalities, i.e., for each $i \in \mathcal{A}$,
$$
v_{i\sigma(i)} - r_{\sigma(i)} = v_{i\sigma'(i)} - r_{\sigma'(i)} \geq v_{ij} - r_j \quad \forall j \in \mathcal{R}.
$$
Thus, $(\sigma', r)$ is also envy-free. It also satisfies the bound constraints because $(\sigma,r)$ does, and we have the same renting here as well.
\end{proof}

\boundsonrents*
\begin{proof}
\label{proof:bounds_on_rents}
    We will prove the above theorem in multiple parts.

\begin{algorithm}[ht]
	\begin{wbox}
		\textsc{BoundsAware-EF-Rents}$(\sigma,r, \mathcal{A}, \mathcal{R}, V, R_t, l, u)$ \\
		\textbf{Input:} A rent division instance with bounds on room rents, $I = (\mathcal{A}, \mathcal{R}, V, R_t, l, u)$, a welfare-maximizing assignment $\sigma$ and some envy-free rent division $r$. \\
        \textbf{Output:} An envy-free rent division $r$ that satisfies the bounds on room rents.
		\begin{enumerate}
		  \item \textbf{Initialization:}
		  \begin{enumerate}
                \item $G_E(\sigma, r) := \text{Envy-Graph}(\sigma, r)$
                \item $L := \{j\in \calR | r_j < l_j \} $
                \item $U := \{j\in \calR | r_j > u_j \} $
                \item $L_f := \{j\in \calR | r_j = l_j \} $
                \item $U_f := \{j\in \calR | r_j = u_j \} $

		  \end{enumerate}  
            \item \textbf{While} $L\cup U \neq \varnothing$ \textbf{do:}
                \begin{enumerate}
                    \item \textbf{If} $U \cap \{ R(L_F) \cup R(L) \} \neq \varnothing $, \textbf{return} ``No solution'';
                    \item \textbf{If} $L \cap \{ R^{-1}(U_F) \cup R^{-1}(U) \} \neq \varnothing $, \textbf{return} ``No solution'';
                    \item Set $S_{inc} := \begin{cases}
                                      R(L) & \text{if $L\neq  \varnothing$ }\\
                                      \calR \setminus \{ R^{-1}(U_F) \cup R^{-1}(U) \} & \text{otherwise}
                                    \end{cases}$
                    \item Set $S_{dec} := \begin{cases}
                                      R^{-1}(U) & \quad\quad\text{if $U\neq  \varnothing$ }\\
                                      \calR \setminus \{ R(L_F) \cup R(L) \} & \quad\quad\text{otherwise}
                                    \end{cases}$
                    \item Increase rents, $r_j$, of all rooms in $S_{inc}$ at unit rate; \\ Decrease rents, $r_j$, of all rooms in $S_{dec}$ at the rate $\frac{|S_{inc}|}{|S_{dec}|}$ simultaneously;
                    \item Stop when a (weak-envy) edge is added to $G_E(r)$ or a vertex moves into or out of $L,U,L_f,U_f$.
                    \item Update $L,U,L_f,U_f \text{ and } G_E(r)$.
            \end{enumerate}
            \item \textbf{Return: $r$} 

		\end{enumerate}
	\end{wbox}
        \caption{Algorithm to find envy-free rents satisfying lower and upper bounds on room rents; $R(X)$ and $R^{-1}(X)$ correspond to the set of vertices that are reachable from, and the set of vertices that can reach some vertex of $X$ in the graph $G_E(r)$ respectively.}
	\label{alg:bounds_on_rents} 
\end{algorithm}

\begin{figure}
    \centering
    \begin{wbox}

\begin{center}
     \begin{tikzpicture}[scale=0.65, >=Stealth, thick, font=\large]
\definecolor{lblue}{RGB}{0,180,255}
\definecolor{lred}{RGB}{255,80,80}
\definecolor{darkgreen}{RGB}{0,130,0}
\begin{scope}[shift={(-5,2)}]
 \draw[lblue, thick] (0,1) circle (1.2);
  \node at (0,1) {$U$};
\end{scope}

\begin{scope}[shift={(5,2)}]
  \draw[lred, thick] (0,1) circle (1.2);
  \node at (0,1) {$L$};
\end{scope}

\node[font=\small, anchor=west] at (-9,4.8) (defU) {$U = \{j \in \mathcal{R} : r_j > u_j\}$};

\node[font=\small, anchor=east] at (9,4.8) (defL) {$L = \{j \in \mathcal{R} : r_j < \ell_j\}$};

\begin{scope}[shift={(-5,2)}]
  \draw[lblue, thick] (0,0) ellipse (2 and 2.6);
 
  \node at (0,-2) {$R^{-1}(U)$};
  \draw[->, lblue, thick] (-1,-0.8) .. controls (-0.4,0.2) .. (-0.4,0.5);
  \draw[->, lblue, thick] (0,-0.8) .. controls (0,0.2) .. (0,0.5);
  \draw[->, lblue, thick] (1,-1) .. controls (0.4,0.2) .. (0.4,0.5);
  \node at (-2.5,1) {\small $r \downarrow$};
\end{scope}

\begin{scope}[shift={(5,2)}]
  \draw[lred, thick] (0,0) ellipse (2 and 2.6);

  \node at (0,-2) {$R(L)$};
  \draw[->, lblue, thick] (-0.4,0.5) .. controls (-0.4,0.2) .. (-1,-0.8);
  \draw[->, lblue, thick] (0,0.5) .. controls (0,0.2) .. (0,-0.8);
  \draw[->, lblue, thick] (0.4,0.5) .. controls (0.4,0.2) .. (1,-1);
  \node at (2.5,1) {\small $r \uparrow$};
\end{scope}

\begin{scope}[shift={(-3,-3)}]
 \draw[black, thick] (0,0.9) circle (0.9);
  \node at (0,1) {$U_f$};
\end{scope}

\begin{scope}[shift={(3,-3)}]
  \draw[black, thick] (0,0.9) circle (0.9);
  \node at (0,1) {$L_f$};
\end{scope}

\node[font=\small, anchor=west] at (-8.9,-3) (defUf) {$U_f = \{j : r_j = u_j\}$};

\node[font=\small, anchor=east] at (8.9,-3) (defLf) {$L_f = \{j : r_j = \ell_j\}$};

\begin{scope}[shift={(-3,-3)}]
  \draw[black, thick] (0,0) ellipse (1.6 and 1.9);
  \node at (0,-1.3) {$R^{-1}(U_f)$};
 \draw[->, lblue, thick] (-0.8,-0.6) .. controls (-0.3,0.2) .. (-0.3,0.5);
  \draw[->, lblue, thick] (0,-0.6) .. controls (0,0.2) .. (0,0.5);
  \draw[->, lblue, thick] (0.8,-0.6) .. controls (0.3,0.2) .. (0.3,0.5);
\end{scope}

\begin{scope}[shift={(3,-3)}]
  \draw[black, thick] (0,0) ellipse (1.6 and 1.9);
  \node at (0,-1.3) {$R^(L_f)$};
   \draw[->, lblue, thick] (-0.3,0.5) .. controls (-0.3,0.2) .. (-0.8,-0.6);
  \draw[->, lblue, thick] (0,0.5) .. controls (0,0.2) .. (0,-0.6);
  \draw[->, lblue, thick] (0.3,0.5) .. controls (0.3,0.2) .. (0.8,-0.6);
 
\end{scope}

 \draw[darkgreen, thick] (0,3) ellipse (1.5 and 1.3);
 \node[darkgreen] at (0,3) {\large rest};
 \node at (0,5) {\small $r \downarrow \text{or} \uparrow$};

\label{fig:algo_bounds}

\end{tikzpicture}
\end{center}

\end{wbox}

    \caption{Figure describing \Cref{alg:bounds_on_rents}}
\end{figure}

\begin{lemma}
\label{lem:bounds_no_solution}
    If \Cref{alg:bounds_on_rents} returns ``No solution'', then there exists no envy-free rents that satisfies the given bounds.
\end{lemma}
\begin{proof}

Assume that the rents computed by the algorithm before returning ``No solution'' is denoted by $r$, and let $r'$ be an envy-free (EF) solution that satisfies the bounds.

When the algorithm returns ``No solution'', one of the following two events must have occurred:
\begin{enumerate}
    \item $U \cap \left( R(L_F) \cup R(L) \right) \neq \varnothing$
    \item $L \cap \left( R^{-1}(U_F) \cup R^{-1}(U) \right) \neq \varnothing$
\end{enumerate}

In case (1), there exists a room $j \in L_F \cup L$ and a room $k \in U$ such that there is a path from $j$ to $k$ in the weak envy graph under rent $r$. Let this path be $j, h_1, h_2, \dots, h_t, k$.

Since $r'$ is an EF rent that satisfies the bounds, we have $r'_k < r_k$. Because there is a weak envy edge from $h_t$ to $k$, it follows that $r'_{h_t} < r_{h_t}$. Continuing this reasoning along the path, we obtain $r'_j < r_j$.

However, since $j \in L_F \cup L$, we know that $r_j \leq \ell_j$, and thus $r'_j < \ell_j$, which contradicts our assumption that $r'$ satisfies the bounds.

An analogous argument applies for case (2), completing the proof.    
\end{proof}

\begin{lemma}
    If \Cref{alg:bounds_on_rents} returns a rent $r$, then $r$ is an envy-free rent satisfying the bounds on rents.
\end{lemma}
\begin{proof}
    The algorithm begins with an envy-free rent vector $r$ that satisfies the total rent constraint. Throughout the while loop, the rents are adjusted in a way that preserves the total rent. Whenever a new weak-envy edge appears, the relevant sets are updated to ensure that no strong envy is introduced. The algorithm terminates and returns $r$ only when the while loop condition fails—that is, when all room rents lie within their respective bounds. Therefore, the returned rent vector $r$ is envy-free, respects the total rent constraint, and satisfies the given lower and upper bounds.

\end{proof}

\begin{lemma}
    The \Cref{alg:bounds_on_rents} terminates.
\end{lemma}
\begin{proof}
We show that after $O(n)$ iterations of the while loop, $|U \cup L|$ will decrease.  
Let's consider two cases:
\begin{enumerate}
    \item $U \neq \emptyset$ and $L \neq \emptyset$
    \item $U = \emptyset$ or $L = \emptyset$
\end{enumerate}

In the first case, we decrease the rents of $R^{-1}(U)$ and increase the rent of $R(L)$, so new edges will be coming into $R^{-1}(U)$ or going out of $R(L)$. Whenever a new edge is introduced, it will increase $|R^{-1}(U) \cup R(L)|$ or move a vertex from $U$ to $U_f$ or from $L$ to $L_f$. Since we only have $n$ vertices, it will decrease $|L \cup U|$ after at most $n$ iterations of the while loop.\\
For the second case, let's say $U \neq \emptyset$ and $L = \emptyset$. The algorithm will decrease the rent of $R^{-1}(U)$ and increase the rent of $\mathcal{R}\setminus (R^{-1}(U_f)\cup R^{-1}(U))$. Let's first do these changes at a very small rate, so we only move all $L_f$ out of it, so let's assume that $L_f = \emptyset$. After each step of the algorithm, either a new edge will come out from the $\mathcal{R}\setminus (R^{-1}(U_f)\cup R^{-1}(U))$ or a new edge will go in the $R^{-1}(U)$. Hence, either $|R^{-1}(U)|$ or $|(R^{-1}(U_f)\cup R^{-1}(U))|$ increases, which again can only happen at most $n$ times after which $|L \cup R|$ must decrease.
As $|U \cup L|$ can only decrease $n$ times, the algorithm will terminate in at most $n^2$ iterations of the while loop. 
\end{proof}
\end{proof}

\maximin*
\begin{proof}
\label{proof:maximin}
Assume the algorithm terminates with output $r$ with $\sigma$ as the associated matching. Suppose, for contradiction, that $(\sigma, r)$ does not maximize the minimum utility, and there exists another envy-free allocation $(\mu, r')$ that does, achieving a higher minimum utility.

By \Cref{thm:welfare}, since both $\sigma$ and $\mu$ are welfare-maximizing matchings, $(\sigma, r')$ is also envy-free. Then, by \Cref{cor:equal_utility}, the utilities of all agents under $(\mu, r')$ and $(\sigma, r')$ are identical. Hence, $(\sigma, r')$ achieves the same (higher) minimum utility as $(\mu, r')$, and so the minimum utility under $(\sigma, r')$ is strictly higher than under $(\sigma, r)$.

Now consider the set $M$ of rooms whose occupants have the minimum utility at the end of the algorithm. Since $r'$ achieves a higher minimum utility, every agent assigned to a room in $M$ must receive strictly more utility under $r'$ than under $r$. Note that the algorithm terminates under one of two conditions:
\begin{enumerate}
    \item $R(L_f) \cap M \neq \varnothing$, or
    \item $R^{-1}(M) \cup R^{-1}(U_f) = \mathcal{R}$.
\end{enumerate}

\textbf{Case 1:} $R(L_f) \cap M \neq \varnothing$. \\
There exists a room $j \in L_f$ and a room $k \in M$ such that there is a directed path from $j$ to $k$ in the envy graph $G_E(\sigma, r)$:
$$
j = h_1 \to h_2 \to \cdots \to h_t = k.
$$
Since $(\sigma, r')$ is envy-free and provides higher utility to the agent in room $k$, it must be that $r'_k < r_k$. To avoid introducing strong envy from the agent in $h_{t-1}$ to $k$, we must have $r'_{h_{t-1}} < r_{h_{t-1}}$. Continuing this argument along the path, we eventually get $r'_j < r_j$. But since $j \in L_f$, we have $r_j = l_j$, and any further decrease would violate the lower bound. Thus, $r'_j$ is infeasible—a contradiction.

\textbf{Case 2:} $R^{-1}(M) \cup R^{-1}(U_f) = \mathcal{R}$. \\
This implies that every agent can reach some room in $M \cup U_f$ via a directed path in $G_E(\sigma, r)$. Using the same logic as in Case 1, we conclude that $r'_j \leq r_j$ for all $j$, and $r'_j < r_j$ for all $j \in M$. But this contradicts the fact that the total rent is same in both $r$ and $r'$ , as we have strictly decreased the rent of some rooms and weakly decreased the rest.

Therefore, our initial assumption is false, and $(\sigma, r)$ must be a maximin utility solution.

To prove that the algorithm terminates, we argue that $|R^{-1}(M) \cup R^{-1}(U_f)|$ increases after at most $n$ iterations of the while loop. Note that when a new edge is created in the algorithm, it should either go into $R^{-1}(M)$ or out of $\calR \setminus (R^{-1}(M) \cup R^{-1}(U_f))$. Either way, we will have a new edge going into $(R^{-1}(M) \cup R^{-1}(U_f))$ and so, the set of vertices that can reach $(R^{-1}(M)$ or $R^{-1}(U_f)$ only increase. Since this can happen only $O(n)$ times, the algorithm terminates within $O(n)$ iterations of the while loop. 





\end{proof}

\leximin*
\begin{proof}
\label{proof:leximin}

Assume the algorithm starts with matching $\sigma$ and terminates with the output $r$ but say $r'(\neq r)$ is the leximin utility rents. Let's say iteration $l$ is the first time the algorithm fixes the rent of a room in $r$ to a value that is different to that of $r'$ by either adding it to $F$ or exiting the while loop. Note that all the rooms whose rents match in $r$ and $r'$ and have been frozen previously are in $F$. Like in \Cref{thm:maximin} are two cases how iteration $l$ ends. The proofs of both these cases very much mimic \Cref{thm:maximin} and we explain them briefly.

\textbf{Case 1:} $R(L_f) \cap M \neq \varnothing$. \\
There exists a room $j \in L_f$ and a room $k \in M$ such that there is a directed path from $j$ to $k$ in the envy graph $G_E(\sigma, r)$:
$
j = h_1 \to h_2 \to \cdots \to h_t = k.
$
Since $(\sigma, r')$ is envy-free and provides higher utility to the agent in room $k$, it must be that $r'_k < r_k$. To avoid introducing strong envy from the agent in $h_{t-1}$ to $k$, we must have $r'_{h_{t-1}} < r_{h_{t-1}}$. Continuing this argument along the path, we eventually get $r'_j < r_j$. But since $j \in L_f$, we have $r_j = l_j$ or $r_j = r'_j$, and any further decrease would violate either one of these necessary conditions giving a contradiction.

\textbf{Case 2:} $R^{-1}(M) \cup R^{-1}(U_f) = \mathcal{R}$. \\
This implies that every agent can reach some room in $M \cup U_f$ via a directed path in $G_E(\sigma, r)$. Using the same logic as in Case 1, we conclude that $r'_j \leq r_j$ for all $j$, and $r'_j < r_j$ for all $j \in M$. But this contradicts the fact that the total rent is the same in both $r$ and $r'$ , as we have strictly decreased the rent of some rooms and weakly decreased the rest. Since neither case is possible, $(\sigma,r)$ is the leximin envy-free solution.

Proof of termination of this algorithm proceeds in the same way as \Cref{alg:maximin}. 

To prove that the algorithm terminates we argue that $|R^{-1}(M) \cup R^{-1}(U_f)|$ increases after at most $n$ iterations of the while loop before $|F|$ increases. Note that when a new edge is created in the algorithm, it should either go into $R^{-1}(M)$ or out of $\calR \setminus (R^{-1}(M) \cup R^{-1}(U_f))$. Either way, we will have a new edge going into $(R^{-1}(M) \cup R^{-1}(U_f))$ and so, The set of vertices that can reach $(R^{-1}(M)$ or $R^{-1}(U_f)$ only increase. Since, this can happen only $O(n)$ times, the size of $|F|$ increases within $O(n)$ iterations of the while loop. And since $|F|$ can only increase from 0 to $n$, the algorithm terminates in $O(n^2)$ iterations of the while loop.

\end{proof}

\minabsolutespread*
\begin{proof}
\label{proof:min_absolute_spread}
The argument to show that the algorithm terminates is very similar to \Cref{alg:maximin} and \Cref{alg:leximin} and is omitted for brevity.

First, note that $r$ is a maximin solution. Assume the algorithm starts with matching $\sigma$ and terminates with the output $r$. Note that $r$ is a maximin utility rent vector - after computing the maximin solution, the algorithm freezes the rents of these rooms and only changes the rents of other rooms. Let $ r'(\neq r)$ be envy-free rents satisfying the budgets with a lesser spread than $r$. AS $r$ is a maximin utility rent vector, the minimum utility of agents in $r'$ cannot be higher than in $r$. And since $r'$ has a smaller spread than $r$, the maximum utility of an agent in $r'$ must be less than that in $r$.

Now consider the set $M$ of rooms whose occupants have the maximum utility(under $r$) at the end of the algorithm. Every agent assigned to a room in $M$ must receive less utility under $r'$ than under $r$. Note that the algorithm terminates under one of two conditions:
\begin{enumerate}
    \item $R^{-1}(U_f) \cap M \neq \varnothing$, or
    \item $R(M) \cup R(L_f) = \mathcal{R}$.
\end{enumerate}

\textbf{Case 1:} $R^{-1}(U_f) \cap M \neq \varnothing$. \\
There exists a room $k \in U_f$ and a room $j \in M$ such that there is a directed path from $j$ to $k$ in the envy graph $G_E(\sigma, r)$:
$$
j = h_1 \to h_2 \to \cdots \to h_t = k.
$$
Since $(\sigma, r')$ is envy-free and provides lower utility to the agents in $M$, it must be that $r'_j > r_j$, i.e., rent of room $j$ must go up from $r$ to $r'$. To avoid introducing strong envy from the agent in $j$ to $h_1$, the rent of $h_1$ should go up by at least that of $j$. So, we have $r'_{h_1} - r_{h_1} \geq r'_j - r_j > 0$.
Continuing this argument along the path, we get $r'_k - r_k \geq r'_j - r_j > 0$.

But since $k \in U_f$, we have either $r_k = u_k$ or $k$ is a room with minimum utility in $r$. If $r_k=u_k$, then $r'_k > r_k = u_k$ making $r'$ infeasible. If $k$ is a room with minimum utility in $r$, then $r'_k - r_k \geq r'_j - r_j \Leftrightarrow u'_\sigma(j)-u'_\sigma(k) \geq u_\sigma(j)-u_\sigma(k)$. Note that $u_\sigma(j)-u_\sigma(k)$ is just the spread of $u$ as $j\in M$ and $k$ are the rooms with maximum and minimum utility respectively.

And so, $r'$ is either infeasible or has at most the spread of $r$—a contradiction.

\textbf{Case 2:} $R(M) \cup R(L_f) = \mathcal{R}$. \\
This implies that every room is reachable from some room in $M \cup L_f$ via a directed path in $G_E(\sigma, r)$. Using the same logic as in Case 1 for every of these paths, we get that $ \forall k\in \calR, j\in M, r'_k - r_k \geq r'_j - r_j$, and $r'_j > r_j$ for all $j \in M$. But this contradicts the fact that the total rent is same in both $r$ and $r'$ , as we have strictly increased the rent of all rooms from $r$ to $r'$.

Therefore, our initial assumption is false, and $(\sigma, r)$ must be a min-spread utility envy-free rent solution.    
\end{proof}

\minrelativespread*
\begin{proof}
\label{proof:min_relative_spread}
By \cref{thm:min_absolute_spread} we know the algorithm returns envy-free rents that minimize the absolute spread, that is also the maximin solution. We claim that these two imply minimum relative spread utility.\\
Assume that $(\sigma,r)$ is the output of the \cref{alg:min_spread}, and agents $a,b$ have minimum and maximum utilities, respectively. Let's say $(\sigma' , r')$ is the min relative spread utility solution, and agents $a',b'$ have minimum and maximum utilities, respectively. \\
By abusing the notation let's call $$u_a= u_{a\sigma(a)(r)},\quad u_b=u_{b\sigma(b)(r)},\quad u_{a'} = u_{a'\sigma'(a')(r')}, \quad u_{b'}=u_{b'\sigma'(b')(r')}$$

We know $(\sigma,r)$ has maximin utility and so, $u_a \geq u_{a'}$. Since, $(\sigma,r)$ has min-spread utility, we have $u_b-u_a \leq u_{b'}-u_{a'}$. Therefore, $\frac{u_b-u_a}{u_a}\leq \frac{u_{b'}-u_{a'}}{u_{a'}}$ 

$$\frac{u_b-u_a}{u_a}\leq \frac{u_{b'}-u_{a'}}{u_{a'}} \Leftrightarrow \frac{u_b}{u_a}\leq \frac{u_{b'}}{u_{a'}} $$

So, the relative spread of the maximum and minimum utility agents under vector $r$ is always at most that of the same agents under vector $r'$. Since, $r'$ has the minimum relative spread, so does $r$, completing the proof.
\end{proof}

\end{document}